\documentclass[12pt]{iopart}

\usepackage{iopams}
\usepackage{amsthm, latexsym, setstack, setstack3, amsfonts, amssymb, graphicx, subfigure}
\usepackage{tikz}
\usetikzlibrary{calc, through, arrows}

\graphicspath{{home/jacob.farinholt/Documents/LaTeX_Files/}{./TeX_Figs/}}
\DeclareGraphicsExtensions{.png} 

\newcommand{\bcup}{\bigcup}

\newcommand{\mc}{\mathcal}
\newcommand{\mbb}{\mathbb}

\newcommand{\ra}{\rangle}
\newcommand{\la}{\langle}

\newcommand{\re}{\textup{Re}}
\newcommand{\im}{\textup{Im}}

\newcommand{\aff}{\text{Aff}}
\newcommand{\Span}{\text{Span}}
\newcommand{\pgamma}{|\gamma\ra\la\gamma|}
\newcommand{\ppsi}{|\psi\ra\la\psi|}
\newcommand{\pvarphi}{|\varphi\ra\la\varphi|}

\newcommand{\id}{\frac{1}{2}\mbb{I}}

\newtheorem{theorem}{Theorem}

\newtheorem{proposition}[theorem]{Proposition}

\newtheorem{lemma}[theorem]{Lemma}
\newtheorem{conjecture}[theorem]{Conjecture}

\begin{document}

\title[The Geometry of Qubit Weak Values]{The Geometry of Qubit Weak Values}
\author{J.~M.~Farinholt}
\address{Strategic \& Computing Systems Department, \\Naval Surface Warfare Ctr, Dahlgren Division, \\Dahlgren, VA 22448}
\ead{Jacob.Farinholt@navy.mil}
\author{A.~Ghazarians}
\address{University of California Berkeley \\Berkeley, CA 94704}
\ead{alanghaza@berkeley.edu}
\author{J.~E.~Troupe}
\address{Applied Research Laboratory,\\University of Texas at Austin, \\Austin, TX 78713}
\ead{jtroupe@arlut.utexas.edu}

\begin{abstract}
The concept of a \emph{weak value} of a quantum observable was developed in the late 1980s by Aharonov and colleagues to characterize the value of an observable for a quantum system in the time interval between two projective measurements. Curiously, these values often lie outside the eigenspectrum of the observable, and can even be complex-valued. Nevertheless, the weak value of a quantum observable has been shown to be a valuable resource in quantum metrology, and has received recent attention in foundational aspects of quantum mechanics. This paper is driven by a desire to more fully understand the underlying mathematical structure of weak values. In order to do this, we allow an observable to be \emph{any} Hermitian operator, and use the pre- and post-selected states to develop well-defined linear maps between the Hermitian operators and their corresponding weak values. We may then use the inherent Euclidean structure on Hermitian space to geometrically decompose a weak value of an observable. In the case in which the quantum systems are qubits, we provide a full geometric characterization of weak values.
\end{abstract}

\submitto{\JPA}

\maketitle
\section{Introduction}

In the late 1980s, Aharonov and colleagues defined what is known as the \emph{weak value} of a quantum mechanical observable \cite{AV90}. The weak value was introduced as a way of representing the value of an observable for a quantum system in the time interval between two projective measurements. The reason for considering such an object is that, unlike classical mechanics where the state of a system is determined completely for all future (and past) times by an initial state and the equations of motions, quantum mechanics allows a degree of freedom with respect to the outcomes of future measurements independent of the equations of motion. Therefore, weak values were introduced as a way to probe this uniquely quantum feature. Of course, in order for weak values to be of physical significance one must in principle be able to experimentally access them. The fundamental principle of measurement disturbance in quantum mechanics requires us to access weak values by performing sufficiently weak measurements on many copies of identical systems having the same initial and final states. The weak value is then proportional to the average weak measurement result of such a pre-selected and post-selected (PPS) ensemble of quantum systems, with the proportionality given by the scale of the coupling strength for the measurement interaction. Fundamentally, the weakness of the measurement is related to the amount of information about the observable that can be extracted from the single measurement interaction. While a single weak measurement cannot extract very much information by itself, a high precision estimate of the weak value can be made with a large PPS ensemble and independent weak measurements on each. 

Since their introduction, weak values have been the focus of a significant amount of interest in quantum metrology and foundations. In metrology, many experiments have used weak measurements and weak values to simplify high precision measurements \cite{Dixon09, Lyons15, Viza15}, including the first observation of the optical spin-Hall effect \cite{Hosten08}.  In foundations, recently the connection between the properties of weak values and so called "quantum paradoxes" has been clarified by Pusey and Leifer \cite{PuseyLeifer15}. In particular, it has been shown that observing anomalous weak values of certain observables demonstrates that any ontological (hidden variable) model of the quantum system must be contextual \cite{Pus14}. The connection between weak values and contextuality was recently used to design a new single qubit quantum key distribution protocol that is fundamentally resistant to some important side-channel attacks \cite{TF15}. Also, since it is a widely held conjecture that contextuality is a key resource for universal quantum computation \cite{Delfosse15, Howard14}, the ability to operationally measure the contextuality of a quantum system using weak measurements to estimate the weak value of selected observables holds potential as a very useful tool for implementing fault-tolerant universal quantum computation. Therefore, understanding the mathematical structure of weak values potentially holds great benefit for more deeply understanding the nature of quantum information and computation. 

In this article we will study the geometrical interplay between between the pre- and post-selected states of the quantum system and the observable being weakly measured, and will relate the geometry to the weak values of the observable. A particular focus will be on the requirements for the weak value to be entirely real or imaginary, and the effect of noise applied to the pre-selected state on the weak value. This is not the first paper to investigate geometric aspects of weak values (see e.g. \cite{Lobo14, Tam09} and references therein). In \cite{Lobo14}, the focus is on the geometric aspects of the system being measured, whereas \cite{Tam09} focuses on the geometry of the measurement observable. This paper, like \cite{Tam09}, is primarily focused on the geometry of the measurement observable. However, it is more general, in that in our case, the space of observables is the set of \emph{all} Hermitian operators, with the state projectors forming a special case. We believe that this approach provides a more complete picture from which to characterize and understand the mathematical framework underlying weak values.

\section{Weak Measurements and Weak Values\label{Sec:WeakMeasurements}}

Suppose the pre-selected state is given by $|\varphi\ra$ and post-selected state is given by $|\psi\ra$ (the ordered pair $(|\varphi\ra, |\psi\ra )$ will henceforth be referred to as the \emph{PPS ensemble}). The corresponding \emph{weak value} of an observable $\widehat{A}$ is given by
\begin{equation}\label{Eq:A_weak}
A_w = \frac{\la\psi|\widehat{A}|\varphi\ra}{\la\psi|\varphi\ra}.
\end{equation}

Contrary to strong measurement results, which always return an eigenvalue of the observable being measured, a PPS ensemble can be chosen such the the corresponding weak value of a nontrivial observable can be any complex number. This seemingly odd behavior of weak measurements drives us to more rigorously characterize the underlying mathematical structure of these measurements, in the hopes of gaining further insight into their nature and potential utility in physical systems.

To that end, we will assume that the quantum systems being weakly measured are finite-dimensional, so that the space of observables $\mbb{H}$ is the set of $n \times n$ Hermitian matrices, where $n \geq 2$. This set is closed under matrix addition and scalar multiplication over $\mbb{R}$, so that $\mbb{H}$ is then an $n^2$-dimensional real vector space.

Now, for any ordered pair $(|\varphi\ra, |\psi\ra)$ of distinct, nonorthogonal pure states, let $W_{\varphi, \psi} : \mbb{H} \rightarrow \mbb{C}$ be a complex linear function defined over $\mbb{H}$, referred to hereinafter as a \emph{weak function}, given by
\begin{equation}\label{Eq:WeakDef}
W_{\varphi, \psi}(M) := \frac{\la\psi|M|\varphi\ra}{\la\psi|\varphi\ra}.
\end{equation}
We choose the notation $W_{\varphi, \psi}(M)$ to denote the weak value of an observable rather than the standard notation of Eq. \eref{Eq:A_weak} from the physics literature for clarity and in order to make its mathematical nature as the image of a linear function more explicit. 

For much of this paper, we will be restricting ourselves to studying a particular subspace of observables, namely the $(n^2 - 1)$-dimensional trace-0 subspace $\mc{S} \subseteq \mbb{H}$, and we will often restrict the weak functions to this domain. In order to make the distinction between weak functions on $\mbb{H}$ and their restriction to the set $\mc{S}$ more explicit throughout the text, we will denote by $\mc{W}_{\varphi, \psi} : \mc{S} \rightarrow \mbb{C}$ the restriction of $W_{\varphi, \psi}$ to $\mc{S}$. Our reason for this restriction is the following.

First of all, the subspace $\mc{S}$ is naturally equipped with a scalar product and corresponding distance that makes it a well-defined Euclidean space. We may consequently analyze various geometric implications this has on the weak functions when restricted to this space. Secondly, if we extend this Euclidean structure over all of $\mbb{H}$, it is easy to see that for any $m \in \mbb{R}$, the set $\mc{A}^{(m)}$ of trace-$m$ Hermitian matrices forms a hyperplane in $\mbb{H}$, and can be viewed as an affine shift of $\mc{S}$. More explicitly, if $\mbb{H}$ denotes the $n \times n$ Hermitian matrices, then we may always write $\mc{A}^{(m)} = \mc{S} + \frac{m}{n}\mbb{I}$, where $\mbb{I}$ is the $n \times n$ identity matrix. It follows that up to the trace of the observable, the underlying structure of weak observables is essentially determined by the trace-0 subspace. For example, suppose $N \in \mbb{H}$ is any $n \times n$ Hermitian matrix, and let $M = N - \frac{\Tr(N)}{n}\mbb{I}$. Then $\Tr(M) = 0$, and
\begin{eqnarray}\label{Eq:Trace_1}
W_{\varphi, \psi}(N) &= \frac{\Tr(N)}{n}W_{\varphi, \psi}(\mbb{I}) + W_{\varphi, \psi}(M)\\ \label{Eq:Trace_2}
 \                   &= \frac{\Tr(N)}{n} + \mc{W}_{\varphi, \psi}(M).
\end{eqnarray}

While weak functions are well-defined over the $n \times n$ Hermitian matrices for any integer $n \geq 2$, for the remainder of this paper, we will restrict ourselves to the base case of $n = 2$. This case corresponds to weak measurements on qubits. Before developing a characterization in full generality, it is important to have a firm understanding of the base case. Indeed, we have found enough rich structure in the qubit case alone, that in the interest of keeping the number of pages of this paper to something even moderately reasonable, we have chosen to reserve the $n > 2$ case for future work.

Recall that, unlike strong measurements, which only return real values as their measurement results, a weak measurement result is in general complex-valued. An interesting question to investigate is what the relationship is between a weak observable and the pre- and post-selected ensemble that might cause the weak value to have a nonzero imaginary component. We will provide a full geometric answer to this question for the qubit case later. We will end this section with an example of a particular case when the imaginary component is necessarily nonzero. This will be used later to prove the more general case, so we state it as a Lemma.

\begin{lemma}\label{Lem:ImaginaryWeak}
 Suppose $|\varphi\ra$ and $|\psi\ra$ are distinct, nonorthogonal qubits, and let $|\gamma\ra$ be a state that is mutually unbiased to both $|\varphi\ra$ and $|\psi\ra$ (that is, $|\la \gamma | \varphi\ra | = |\la \psi | \gamma \ra | = 1/\sqrt{2}$). Then $ \im(W_{\varphi, \psi}( |\gamma\ra\la\gamma|)) \neq 0$.
\begin{proof}
 Because $|\varphi\ra$ and $|\gamma\ra$ are mutually unbiased by assumption, we have that $|\la\gamma|\varphi\ra| = \frac{1}{\sqrt{2}}$. However, with an appropriate global phase rotation applied to $|\varphi\ra$, we may assume without any loss of generality that in fact $\la\gamma|\varphi\ra = \frac{1}{\sqrt{2}}$. Thus, we may write
\begin{equation}\label{Eq:1}
 |\varphi\ra = \frac{1}{\sqrt{2}}(|\gamma\ra + \e^{i \phi}|\gamma_\bot\ra),
\end{equation}
 for some $\phi \in [0, 2\pi)$, where $|\gamma_\bot\ra$ is the orthogonal complement of $|\gamma\ra$. Similarly, since $|\psi\ra$ and $|\gamma\ra$ are also mutually unbiased, we may apply an appropriate global phase shift on $|\psi\ra$ such that $\la\psi|\gamma\ra = \frac{1}{\sqrt{2}}$. Hence, we may write
\begin{equation}\label{Eq:2}
 |\psi\ra = \frac{1}{\sqrt{2}}(|\gamma\ra + \e^{i \alpha}|\gamma_\bot\ra),
\end{equation}
 for some $\alpha \in [0, 2\pi)$.

 Once we have fixed these global phases on $|\varphi\ra$ and $|\psi\ra$ to force the two inner products $\la\psi|\gamma\ra$ and $\la\gamma|\varphi\ra$ to be real, we no longer have the luxury of assuming the inner product $\la\psi|\varphi\ra$ is real. However, using \eref{Eq:1} and \eref{Eq:2}, we calculate
\begin{equation}\label{Eq:InnerProduct}
 \la\psi|\varphi\ra = \frac{1}{2}\left(1 + \e^{i(\phi - \alpha)}\right).
\end{equation}
 It follows that the inner product $\la\psi|\varphi\ra$ is a real number if and only if $(\phi - \alpha) \in \{0, \pi\}$, which occurs if and only if either $|\psi\ra = |\varphi\ra$ or $|\psi\ra$ is orthogonal to $|\varphi\ra$, neither of which are true by assumption. Hence, the weak value
\begin{equation}
 W_{\varphi, \psi}( |\gamma\ra\la\gamma|) = \frac{\la\psi\pgamma \varphi\ra}{\la\psi|\varphi\ra}
\end{equation}
necessarily contains a nonzero imaginary component, which concludes our proof.
\end{proof}
\end{lemma}
 In fact, letting $\omega = (\alpha - \phi)$, we explicitly calculate the weak value to be
\begin{equation}\label{Eq:AwGamma}
 W_{\varphi, \psi}( |\gamma\ra\la\gamma|) = \frac{1}{2}\left(1 + i\tan\left(\frac{\omega}{2}\right)\right).
\end{equation}
In particular, if $|\varphi\ra$ and $|\psi\ra$ are mutually unbiased, then $\omega = \pm \frac{\pi}{2}$, so that the weak value then reduces to $\frac{1}{2}(1 \pm i)$.

Because we will be relying on the construction of Lemma \ref{Lem:ImaginaryWeak} throughout this paper, it will be helpful to standardize some notation. Throughout this paper, the state $|\gamma\ra$ will always refer to a state that is mutually unbiased to both the pre- and post-selected states $|\varphi\ra$ and $|\psi\ra$. In particular, we will always assume $|\varphi\ra$ and $|\psi\ra$ have a decomposition in terms of the $|\gamma\ra$, $|\gamma_\bot\ra$ basis as in Equations \eref{Eq:1} and \eref{Eq:2}, respectively. The phase term $\omega$ will always have the form $\omega = (\alpha - \phi)$, so that we may always rewrite Equation \eref{Eq:InnerProduct} as
\begin{equation}\label{Eq:Omega}
 \la\psi|\varphi\ra = \frac{1}{2}\left(1 + \e^{-i\omega}\right).
\end{equation}

\section{The Geometry of Density Operators\label{Sec:GeomOfDensityMatrices}}

Density operators are defined as the set of positive trace-1 Hermitian operators. These generalize the quantum states to include both pure and mixed states. The pure states map to rank 1 projectors ($|\psi\ra \mapsto \ppsi$), and the mixed states form the convex hull of the projectors. Recall from the previous section that the set $\mc{A}^{(1)}$ of trace-1 Hermitian matrices may be viewed as an affine shift of the vector space $\mc{S}$ of trace-0 Hermitian matrices. Restricting ourselves to the $n=2$ case, we may explicitly write $\mc{A}^{(1)} = \mc{S} + \frac{1}{2}\mbb{I}$. 

As was hinted at earlier, we can define a scalar product and distance over the space $\mc{S}$ that makes it a well-defined Euclidean space. The scalar product of two elements $A$ and $B$ of $\mc{S}$ is given by
\begin{equation}\label{Eq:ScalarProduct}
 (A,B) = \frac{1}{2}(\Tr AB),
\end{equation}
and the corresponding Euclidean distance squared is thus given by
\begin{equation}\label{Eq:Distance}
 D^2(A,B) = \frac{1}{2} \Tr(A-B)^2.
\end{equation}
Because of this added structure on the space $\mc{S}$, it is convenient to study the image $\mc{D} \subseteq \mc{S}$ of the set of density operators under the affine shift $\mc{A}^{(1)} \rightarrow \mc{S}$. Straightforward calculations reveal that the set $\mc{D}$ forms a 3-ball centered around the origin, having radius 1/2. Its surface, corresponding to the set of pure states, is composed entirely of elements of the form $\ppsi - \frac{1}{2}\mbb{I}$.  We will call $\mc{D}$ the \emph{Bloch ball} and its surface the \emph{Bloch sphere}.

It follows from \eref{Eq:ScalarProduct} and \eref{Eq:Distance} that two states are orthogonal if and only if the corresponding points in $\mc{D}$ sit at unit distance from each other. Consequently, orthogonal pure states correspond to antipodal points on the sphere.

Two qubit states $|\varphi\ra$ and $|\psi\ra$ are called \emph{mutually unbiased} if and only if $|\la \varphi |\psi\ra|^2 = \frac{1}{2}$. Straightforward calculations show that mutual unbiasedness of pure states $|\varphi\ra$ and $|\psi\ra$ is equivalent to orthogonality of the corresponding points $(\pvarphi - \frac{1}{2}\mbb{I})$ and $(\ppsi - \frac{1}{2}\mbb{I})$ in $\mc{D}$.

Naturally, any set of three orthogonal points on the surface of $\mc{D}$ forms a basis in $\mc{S}$. That is, if $|\varphi\ra$, $|\psi\ra$ and $|\gamma\ra$ are all three mutually unbiased, then any $2 \times 2$ trace-0 Hermitian operator $M$ can be uniquely represented as
\begin{equation}\label{Eq:Hermitian}
 M = a|\varphi\ra\la\varphi| + b|\psi\ra\la\psi| + c|\gamma\ra\la\gamma| - \frac{a+b+c}{2}\mbb{I},
\end{equation}
for $a, b, c \in \mbb{R}$. In particular, $M$ is an element of the Bloch ball if and only if $\sqrt{a^2 + b^2 + c^2} \leq 1$. Admittedly, this representation of the Bloch ball is a bit non-standard. However, if we let $|\varphi\ra$, $|\psi\ra$, and $|\gamma\ra$ be given by the $+1$ eigenstates of the Pauli spin $X$, $Y$, and $Z$ matrices, respectively, then one may show that $M$ is equivalently written as
\begin{equation}\label{Eq:Bloch}
 M = \frac{1}{2}(aX +bY + cZ),
\end{equation}
where $a$, $b$, and $c$ are the same exact values as in \eref{Eq:Hermitian}. Consequently, shifting $M$ back into the trace-1 space $M \mapsto M + \frac{1}{2}\mbb{I}$ gives us the more familiar representation
\begin{equation}
\frac{1}{2}(\mbb{I} + aX + bY + cZ)
\end{equation}
of a density operator. However, for our sake, it will be far more convenient to both remain in the trace-0 subspace $\mc{S}$ and avoid fixing a particular basis.

\section{The Geometry of Weak Values}\label{Sec:GeomWeakAll}

Earlier in this paper, we asked what the relationship is between a Hermitian matrix $M \in \mbb{H}$ (i.e. a weak observable) and a PPS ensemble $(|\varphi\ra, |\psi\ra)$ that might cause the corresponding weak value $W_{\varphi, \psi}(M)$ to have a nonzero imaginary component. We are now ready to address this question more formally in the form of a theorem, and it will serve as the motivation for the remainder of this section. The proof, in fact, will be given later in this section, after we have developed some of the necessary geometric tools. Before stating the theorem, we introduce some notation.

If $|\varphi\ra$ and $|\psi\ra$ are distinct, non-orthogonal pure states, then the span of their images in $\mc{S}$, $\Span_\mc{S}\{\pvarphi - \frac{1}{2}\mbb{I}, \ \ppsi - \frac{1}{2}\mbb{I}\}$, is a two-dimensional subspace, that is, a plane, in $\mc{S}$. If the weak value function we are characterizing is given by $\mc{W}_{\varphi, \psi}$ or $\mc{W}_{\psi, \varphi}$, then we call $\Span_\mc{S}\{\pvarphi - \frac{1}{2}\mbb{I}, \ \ppsi - \frac{1}{2}\mbb{I}\}$ the \emph{PPS plane}, and denote it by $H_{pps}$. Figure \ref{Fig:CirclePlot}(a) depicts the intersection of the PPS plane corresponding to some weak value function with the Bloch sphere.

\begin{theorem}\label{Thm:RealWeak}
 Suppose $|\varphi\ra$ and $|\psi\ra$ are distinct non-orthogonal qubits and $M$ is a trace-0 Hermitian operator. Then $\mc{W}_{\varphi, \psi}( M)$ is a strictly real number if and only if $M \in H_{pps}$.
\end{theorem}

The fact that we can obtain necessary and sufficient conditions for a weak value to be real that are purely geometric in nature strongly motivates the possibility that a complete characterization of qubit weak measurements, at least over trace-0 Hermitian matrices, will be fundamentally geometric. In fact, observe that, for any pure states $|\varphi\ra$ and $|\psi\ra$, we trivially have $\mc{W}_{\varphi, \psi}( |\varphi\ra\la\varphi| - \frac{1}{2}\mbb{I}) = \mc{W}_{\varphi, \psi}(|\psi\ra\la\psi| - \frac{1}{2}\mbb{I}) = \frac{1}{2}$. More generally, let $M = (1-a)\ppsi + a\pvarphi - \frac{1}{2}\mbb{I}$ for some $a \in \mbb{R}$. Then $\mc{W}_{\varphi, \psi}(M) = \frac{1}{2}$. Thus, all elements of the affine line $\aff_{\mc{S}}\{\ppsi - \id , \ \pvarphi - \id\}$ map to the same value under the weak function $\mc{W}_{\varphi, \psi}$. With just a few simple calculations, one also sees that any parallel shift of this line along the PPS plane will be mapped to a distinct value under $\mc{W}_{\varphi, \psi}$. Consequently, we may call the line $\aff_{\mc{S}}\{\ppsi - \id , \ \pvarphi - \id\}$ (and any of its parallel shifts along the PPS plane) a \emph{weak value invariant} for the weak function $\mc{W}_{\varphi, \psi}$.  The remainder of this section will be dedicated to characterizing this invariance.

For a given PPS ensemble $(|\varphi\ra, |\psi\ra)$, define the point
\begin{equation}
 P_{\varphi, \psi} := \frac{1}{2}(\pvarphi + \ppsi - \mbb{I}). 
\end{equation}
This is the point exactly half-way between $\pvarphi - \id$ and $\ppsi - \id$ in the PPS plane, and hence the vector $\vec{P}_{\varphi, \psi}$ is perpendicular to $\aff_{\mc{S}}\{\ppsi - \id , \ \pvarphi - \id\}$ in $H_{pps}$. Any parallel shift of the line $\aff_{\mc{S}}\{\ppsi - \id , \ \pvarphi - \id\}$ along the PPS plane will contain some scalar multiple of $P_{\varphi, \psi}$. Thus, we define the function $\mc{K}_{\varphi, \psi}$ that maps a real number to a line in $H_{pps}$ via
\begin{equation}
 \mc{K}_{\varphi, \psi}(s) := \aff_{\mc{S}}\left\{\ppsi - \id , \ \pvarphi - \id\right\} + (s-1)P_{\varphi, \psi}.
\end{equation}

We see that the function $\mc{K}_{\varphi, \psi}$ shifts the line $\aff_{\mc{S}}\{\ppsi - \id , \ \pvarphi - \id\}$ so that it is centered over a designated scalar multiple of the point $P_{\varphi, \psi}$ (see Figure \ref{Fig:CirclePlot}(b)). In particular, $\mc{K}_{\varphi, \psi}(1) = \aff_{\mc{S}}\{\ppsi - \id , \ \pvarphi - \id\}$ and $\mc{K}_{\varphi, \psi}(0)$ contains the origin. Because each one of these lines is a weak value invariant, we may abuse our notation slightly and write instead
\begin{equation}
 \mc{W}_{\varphi, \psi}( \mc{K}_{\varphi, \psi}(s)) := \mc{W}_{\varphi, \psi}( M), \text{ for } M \text{ any element of } \mc{K}_{\varphi, \psi}(s).
\end{equation}
In particular,
\begin{eqnarray}\label{Eq:WeakK}
 \mc{W}_{\varphi, \psi}( \mc{K}_{\varphi, \psi}(s)) &= \mc{W}_{\varphi, \psi}(sP_{\varphi, \psi})\\
                                            \       &= s\mc{W}_{\varphi, \psi}(P_{\varphi, \psi})\\
                                            \       &= \frac{s}{2}.
\end{eqnarray}
In other words, for all $s \in \mbb{R}$, $\mc{K}_{\varphi, \psi}(s)$ is the preimage of $\frac{s}{2}$ under $\mc{W}_{\varphi, \psi}$.

\begin{figure}
\centering 

\begin{subfigure}[]{}%
\includegraphics[width=0.8\textwidth]{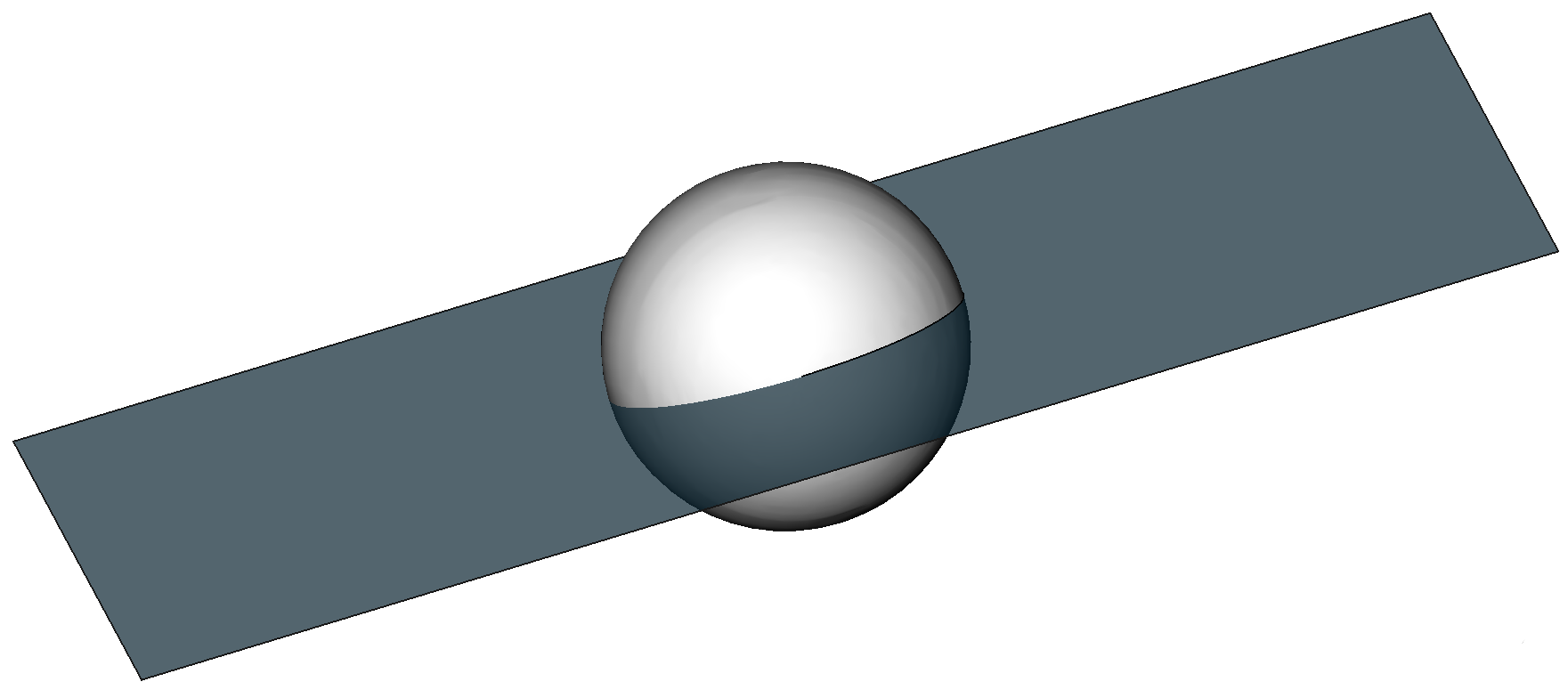}%
\end{subfigure}

\begin{subfigure}[]{}
\begin{tikzpicture}[scale=3]
\coordinate [] (origin) at (0,0);
\coordinate [label=right:$|\psi\ra\la\psi| - \frac{1}{2}\mbb{I}$] (psi) at (1,0);
\coordinate [label=right:$|\varphi\ra\la\varphi| - \frac{1}{2}\mbb{I}$] (phi) at (75:1cm);
\coordinate [label=below left:$P_{\varphi, \psi}$] (P) at ($ (phi)!.5!(psi) $) {};
\coordinate [thick, label=above right:$\mc{K}_{\varphi, \psi}(1)$] (Kbot) at ($ (phi)!2!(psi) $){};
\coordinate (Ktop) at ($ (psi)!1.5!(phi) $){};
\coordinate [label=right:$sP_{\varphi, \psi}$] (sP) at ($ (origin)!2!(P) $) {}; 
\coordinate (s*psi) at ($ (psi)+(P) $) {}; 
\coordinate [thick, label=above right:$\mc{K}_{\varphi, \psi}(s)$] (sKbot) at ($(Kbot)+(P) $){};
\coordinate (sKtop) at ($(Ktop)+(P) $){};

\node (psinode) at (psi) [circle, draw, fill, inner sep=1.5pt]{};
\node (phinode) at (phi) [circle, draw, fill, inner sep=1.5pt]{};
\node (Pnode) at (P) [circle, draw, fill, inner sep=1.5pt]{};
\node (sPnode) at (sP) [circle, draw, fill, inner sep=1.5pt]{};

\node (psipoint) at ($ (origin)!.999!(psi) $) {};
\node (phipoint) at ($ (origin)!.999!(phi) $) {};
\node (Ppoint) at ($ (origin)!.999!(P) $){};
\node (sPpoint) at ($ (origin)!.999!(sP)$){};

\draw (Ktop)--(Kbot);
\draw (sKtop)--(sKbot);
\draw [->, thick] (P)--(sPpoint);

\node (Circ) [draw,circle through=(phi)] at (origin) {};
\end{tikzpicture}
\end{subfigure}
\caption{(a) Intersection of the Bloch Sphere with the PPS plane of a given weak function $\mc{W}_{\varphi, \psi}$. (b) A shift of the line $\mc{K}_{\varphi, \psi}(1) = \aff_{\mc{S}}\{\ppsi - \id , \ \pvarphi - \id\}$ along the PPS plane. The circle denotes the intersection of the Bloch sphere with the PPS plane.}\label{Fig:CirclePlot}%
\end{figure}

It is straightforward to see that
\begin{equation}
 \bcup_{s \in \mbb{R}} \mc{K}_{\varphi, \psi}(s) = H_{pps}.
\end{equation}
Consequently, the sufficiency part of Theorem \ref{Thm:RealWeak} follows trivially. Stated explicitly:
\begin{proposition}\label{Prop:RealWeakSufficiency}
 For a given PPS ensemble $\{|\varphi\ra, |\psi\ra\}$, let $M$ be any Hermitian operator in the corresponding PPS plane $H_{pps}$. Then $\mc{W}_{\varphi, \psi}( M) \in \mbb{R}$.
\end{proposition}

We may now prove Theorem \ref{Thm:RealWeak}.
\begin{proof}[Proof of Theorem \ref{Thm:RealWeak}]
Since sufficiency follows immediately from Proposition \ref{Prop:RealWeakSufficiency}, we need only show necessity.

Let $|\psi'\ra$ be mutually unbiased to $|\varphi\ra$ such that $|\psi\ra\la\psi|- \frac{1}{2}\mbb{I} \in \Span_\mc{S}\{|\varphi\ra\la\varphi|-\frac{1}{2}\mbb{I}, \ |\psi'\ra\la\psi'| - \frac{1}{2}\mbb{I}\}$. Then $\Span_\mc{S}\{|\varphi\ra\la\varphi|-\frac{1}{2}\mbb{I}, \ |\psi'\ra\la\psi'| - \frac{1}{2}\mbb{I}\} = H_{pps}$, and there exists a state $|\gamma\ra$ mutually unbiased to both $|\varphi\ra$ and $|\psi'\ra$, so that we may write an arbitrary trace-0 Hermitian operator $M$ as
\begin{equation}
 M = a\left(|\varphi\ra\la\varphi| - \frac{1}{2}\mbb{I}\right) + b\left(|\psi'\ra\la\psi'| - \frac{1}{2}\mbb{I}\right) + c\left(|\gamma\ra\la\gamma| - \frac{1}{2}\mbb{I}\right),
\end{equation}
for some $a, b, c \in \mbb{R}$. If $c = 0$, then $M \in H_{pps}$, so assume $c \neq 0$. Then by linearity of $\mc{W}_{\varphi, \psi}$, we see that
\begin{eqnarray*}
 \mc{W}_{\varphi, \psi}( M) &= \mc{W}_{\varphi, \psi}\left(a\left(|\varphi\ra\la\varphi| - \frac{1}{2}\mbb{I}\right) + b\left(|\psi'\ra\la\psi'| - \frac{1}{2}\mbb{I}\right)\right)\\
 \ &\ \ \ + c\mc{W}_{\varphi, \psi}\left(|\gamma\ra\la\gamma| - \frac{1}{2}\mbb{I}\right).
\end{eqnarray*}
By Proposition \ref{Prop:RealWeakSufficiency} it follows that $\mc{W}_{\varphi, \psi}\left(a\left(|\varphi\ra\la\varphi| - \frac{1}{2}\mbb{I}\right) + b\left(|\psi'\ra\la\psi'| - \frac{1}{2}\mbb{I}\right)\right)$ is real, since $a\left(|\varphi\ra\la\varphi| - \frac{1}{2}\mbb{I}\right) + b\left(|\psi'\ra\la\psi'| - \frac{1}{2}\mbb{I}\right)$ is in $H_{pps}$. Furthermore, we may write
\begin{eqnarray}
\mc{W}_{\varphi, \psi}\left(|\gamma\ra\la\gamma| - \frac{1}{2}\mbb{I}\right) &= W_{\varphi, \psi}\left(|\gamma\ra\la\gamma|\right) - W_{\varphi, \psi}\left(\frac{1}{2}\mbb{I}\right)\\
 \ &= W_{\varphi, \psi}\left(|\gamma\ra\la\gamma|\right) - \frac{1}{2}.
\end{eqnarray}
Consequently, it suffices to show that $W_{\varphi, \psi}( |\gamma\ra\la\gamma|) \not\in \mbb{R}$. But this follows immediately from Lemma \ref{Lem:ImaginaryWeak}. We conclude that if $M$ is not in $H_{pps}$, then $\mc{W}_{\varphi, \psi}( M) \not\in \mbb{R}$, thus proving the necessity.
\end{proof}

In fact, we may be even more explicit. Let $|\gamma\ra$ be mutually unbiased to both $|\varphi\ra$ and $|\psi\ra$, so that $|\varphi\ra$ and $|\psi\ra$ have the form of \eref{Eq:1} and \eref{Eq:2}, respectively. Let $\mc{K}_{\varphi, \psi}(s,a)$ denote the line $\mc{K}_{\varphi, \psi}(s) + a(\pgamma - \id)$. This is the line $\mc{K}_{\varphi, \psi}(s)$ lifted off of the PPS plane in a perpendicular manner by an amount determined by the scalar $a$. Then using Equations \eref{Eq:AwGamma} and \eref{Eq:WeakK}, we calculate
\begin{equation}
\mc{W}_{\varphi, \psi}( \mc{K}_{\varphi, \psi}(s,a)) = \frac{1}{2}\left(s + i a \tan\left(\frac{\omega}{2}\right)\right).
\end{equation}
Consequently, the imaginary component of the weak value contains two pieces of information. The value $\omega$ is directly related to the angle between $|\varphi\ra$ and $|\psi\ra$, and the value $a$ indicates how far away the line containing the observable is from the PPS plane. In particular, if $|\varphi\ra$ and $|\psi\ra$ are also mutually unbiased to each other, then the weak value reduces to
\begin{equation}
\mc{W}_{\varphi, \psi}( \mc{K}_{\varphi, \psi}(s,a)) = \frac{1}{2}(s \pm i a),
\end{equation}
where the sign of the imaginary component is determined by the sign of $\omega$.

We have thus obtained a full geometric characterization of the weak value of trace-0 qubit weak observables. It follows from Eq. \eref{Eq:Trace_1} and \eref{Eq:Trace_2} that we can completely characterize the weak value of \emph{any} qubit weak observable. Namely, suppose $N$ is any $2 \times 2$ Hermitian operator, and let $M = N - \frac{\Tr(N)}{2}\mbb{I}$. Suppose $\mc{W}_{\varphi, \psi}( M) = \frac{1}{2}(s + i a \tan(\frac{\omega}{2}))$. It follows that 
\begin{equation}\label{Eq:TotalWeakValue}
W_{\varphi, \psi}( N) = \frac{1}{2}\left(\Tr N + s + i a \tan\left(\frac{\omega}{2}\right)\right).
\end{equation}
This also implies that the real part of a weak value contains two pieces of information; namely, it relates to the trace of the weak observable, and the location of the line containing its projection onto the PPS plane.

\section{Weak Values of State Projectors}\label{Sec:WeakValueProj}

From an experimental point of view, perhaps the most important observables are state projectors. Recall that any Hermitian matrix can be written as a linear sum of projectors onto its eigenvectors (where the scalars are given by the corresponding eigenvalues). Thus, for any $2\times 2$ Hermitian matrix $M \in \mbb{H}$, there exists a pure state $|\psi\ra$ with orthogonal complement $|\psi_\bot\ra$ and scalars $a, b \in \mbb{R}$ such that $M$ can be written as 
\begin{equation}
M = a\ppsi + b|\psi_\bot\ra\la\psi_\bot|.
\end{equation}
Consequently, with an appropriate scaling of the measurement results, we can infer all the necessary information about an observable by simply measuring state projectors. Thus, in this section, we will consider weak measurements of observables that are points on the Bloch sphere. Because state projectors all have trace-1, it follows from Eq. \eref{Eq:TotalWeakValue} that we need only add $1/2$ to the resulting weak values in order to obtain the full characterization of the weak functions on state projectors.

Because the Bloch sphere is bounded, so too is its image in the complex plane under any weak function. Thus, we may further consider questions regarding optimization of weak functions over the Bloch sphere. Again taking experimental utility into consideration, we note that it is generally difficult to experimentally obtain both the real and imaginary parts of a weak value from a weak measurement, so we are generally left with choosing one of the two to measure.

Suppose we are interested in only obtaining the real part of a weak value. Note that, for any $a \in \mbb{R}$,  $\re(\mc{W}_{\varphi, \psi}(\mc{K}_{\varphi, \psi}(s,a))) = \frac{s}{2}$. Consequently, we gain nothing regarding the real part of a weak value by considering operators not on the PPS plane. Thus, we may restrict ourselves further to considering weak functions restricted to the intersection of the Bloch sphere with the corresponding PPS plane. This is simply a circle, so we denote this space by $\mc{C}$.

We may now ask the following question: Given a particular PPS ensemble $(|\varphi\ra, |\psi\ra)$, what are $\max_{M \in \mc{C}} \mc{W}_{\varphi, \psi}(M)$ and $\min_{M \in \mc{C}} \mc{W}_{\varphi, \psi}(M)$? These values are obtained by finding the largest and smallest values $s \in \mbb{R}$ such that the line $\mc{K}_{\varphi, \psi}(s) = \mc{K}_{\varphi, \psi}(s, 0)$ intersected with $\mc{C}$ is nonempty. Obviously, the maximum and minimum are reached precisely when $\mc{K}_{\varphi, \psi}(s)$ is tangent to the circle $\mc{C}$. These tangent points will both be multiples of the point $P_{\varphi, \psi}$. Thus, the problem reduces to finding the values $s \in \mbb{R}$ such that the points $sP_{\varphi, \psi}$ are points on the circle, that is, they are distance $1/2$ from the origin. Calculating, we find that $s = \pm 1/|\la \varphi |\psi\ra|$, so that
\begin{equation}
\max_{M \in \mc{C}} \mc{W}_{\varphi, \psi}(M) = \mc{W}_{\varphi, \psi}\left(\mc{K}_{\varphi, \psi}\left(\frac{1}{|\la \varphi |\psi\ra|}\right)\right) = \frac{1}{2|\la \varphi |\psi\ra|},
\end{equation}
and
\begin{equation}
\min_{M \in \mc{C}} \mc{W}_{\varphi, \psi}(M) = \mc{W}_{\varphi, \psi}\left(\mc{K}_{\varphi, \psi}\left(-\frac{1}{|\la \varphi |\psi\ra|}\right)\right) = -\frac{1}{2|\la \varphi |\psi\ra|},
\end{equation}

By shifting back to the trace-1 space, we obtain the following result.

\begin{theorem}
For a given PPS ensemble $(|\varphi\ra, |\psi\ra)$, the state projectors that produce the largest and smallest real weak values are given by 
\begin{equation}\label{Eq:H_MaxMin}
H^{\pm}_{\varphi, \psi} = \frac{1-s}{2}\mbb{I} + \frac{s}{2}(\pvarphi + \ppsi),
\end{equation}
where $s = \pm 1/|\la \varphi |\psi\ra|$, and the superscript of $H_{\varphi, \psi}$ is determined by the sign of $s$. The corresponding weak values are 
\begin{equation}\label{Eq:MaxRealWeak}
W_{\varphi, \psi}(H^\pm_{\varphi, \psi}) = \frac{1}{2}\left(1 \pm \frac{1}{|\la\varphi|\psi\ra|}\right).
\end{equation}
\end{theorem}
The values in \eref{Eq:MaxRealWeak} consequently provide the bounds of obtainable real weak values for a given PPS ensemble when we restrict the set of weak observables to state projectors. Note that, as $|\varphi\ra$ and $|\psi\ra$ become closer to orthogonal, the bounds approach $\pm \infty$.

It is interesting to note that the state projectors $H_{\varphi, \psi}^\pm$ always correspond to antipodal points on the Bloch sphere. Consequently, the states onto which they project are orthogonal to each other. The eigenspectrum of a state projector is simply $\{0, 1\}$, and hence $H^+_{\varphi, \psi}$ and $H^-_{\varphi, \psi}$ have the same amplification effect when acted on by $W_{\varphi, \psi}$, just in opposite directions. If the amplification effect is the resource for a particular application, then both $H^+_{\varphi, \psi}$ and $H^-_{\varphi, \psi}$ are optimal.

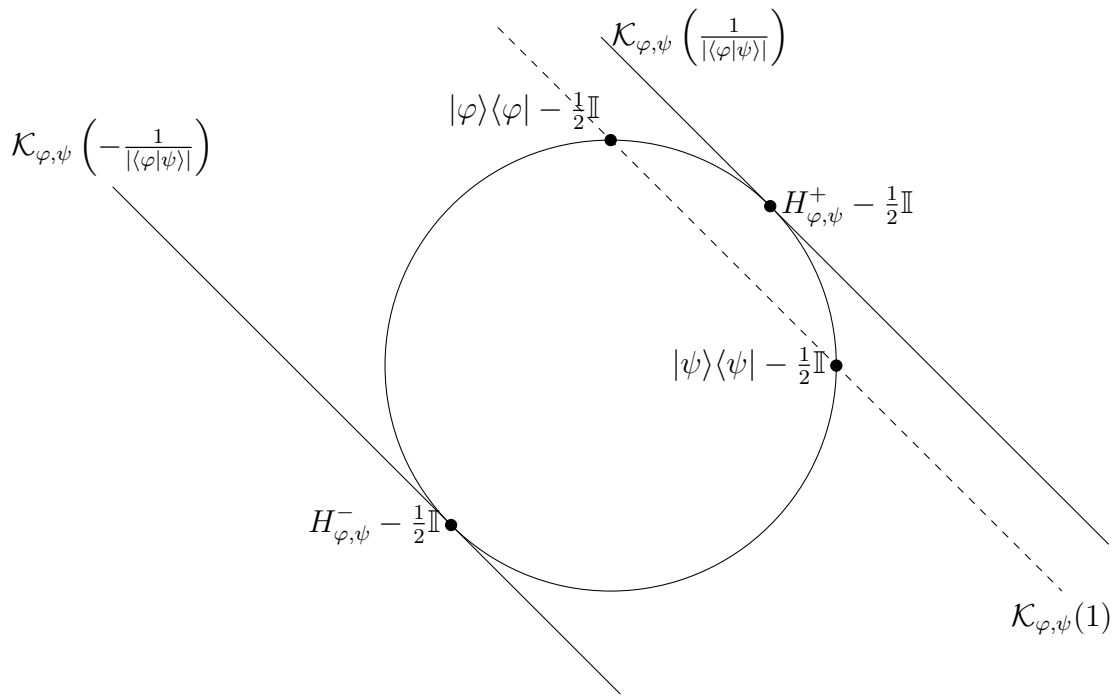
\begin{figure}
\centering 
\begin{tikzpicture}[scale=3]
\coordinate [] (origin) at (0,0);
\coordinate [label=left:$|\psi\ra\la\psi| - \frac{1}{2}\mbb{I}$] (psi) at (1,0);
\coordinate [label=above left:$|\varphi\ra\la\varphi| - \frac{1}{2}\mbb{I}$] (phi) at (90:1cm);
\coordinate (P) at ($ (phi)!.5!(psi) $) {};
\coordinate [label=right:$H^+_{\varphi, \psi} - \frac{1}{2}\mbb{I}$] (sP) at (45:1cm) {}; 
\coordinate (sPhi) at ($ (phi)+(sP)-(P) $) {};
\coordinate [thick, label=right:$\mc{K}_{\varphi, \psi}\left(\frac{1}{|\la\varphi|\psi\ra|}\right)$] (sKtop) at ($ (sP)!1.5!(sPhi) $){};
\coordinate (sPsi) at ($ (psi)+(sP)-(P) $) {}; 
\coordinate  (sKbot) at ($ (sP)!3!(sPsi) $){};

\coordinate [label=left:$H^-_{\varphi, \psi} - \frac{1}{2}\mbb{I}$] (-sP) at (225:1cm) {};
\coordinate (-sPhi) at ($ (phi)+(-sP)-(P) $) {};
\coordinate [thick, label=above:$\mc{K}_{\varphi, \psi}\left(-\frac{1}{|\la\varphi|\psi\ra|}\right)$](-sKtop) at ($ (-sP)!3!(-sPhi) $){};
\coordinate (-sPsi) at ($ (psi)+(-sP)-(P) $) {}; 
\coordinate  (-sKbot) at ($ (-sP)!1.5!(-sPsi) $){};

\coordinate [thick, label=below:$\mc{K}_{\varphi, \psi}(1)$] (Kbot) at ($ (phi)!2!(psi) $){};
\coordinate (Ktop) at ($ (psi)!1.5!(phi) $){};

\node (psinode) at (psi) [circle, draw, fill, inner sep=1.5pt]{};
\node (phinode) at (phi) [circle, draw, fill, inner sep=1.5pt]{};
\node (sPnode) at (sP) [circle, draw, fill, inner sep=1.5pt]{};
\node (sPnode) at (-sP) [circle, draw, fill, inner sep=1.5pt]{};

\draw (Ktop)--(Kbot)[dashed];
\draw (sKtop)--(sKbot);
\draw (-sKtop)--(-sKbot);

\node (Circ) [draw,circle through=(phi)] at (origin) {};

\end{tikzpicture}%
\caption{Shifting the line $\mc{K}_{\varphi, \psi}(s)$ along the PPS plane until it is tangent to the circle provides a simple method of finding the state projector that maximizes or minimizes the corresponding real weak value.}\label{Fig:CirclePlot2}%
\end{figure}

A similar geometric approach can be used to characterize the projectors that maximize and minimize the imaginary component of the weak value. For a given PPS ensemble $(|\varphi\ra, |\psi\ra)$, choose $|\gamma\ra$ to be mutually unbiased to both $|\varphi\ra$ and $|\psi\ra$, and define $\omega$ accordingly as in  Eq. \eref{Eq:Omega}. Observe that, for any $s \in \mbb{R}$, $\im(\mc{W}_{\varphi, \psi}(\mc{K}_{\varphi, \psi}(s,a))) = \frac{a}{2}\tan\left(\frac{\omega}{2}\right)$. Consequently, the projectors that maximize and minimize the imaginary component of the weak value are those whose images in the trace-0 subspace are farthest away from the PPS plane. Thus, we obtain the following result.

\begin{theorem}
For a given PPS ensemble $(\varphi\ra, |\psi\ra )$, let $|\gamma\ra$ be a state mutually unbiased to both $|\varphi\ra$ and $|\psi\ra$. Then the state projectors whose imaginary components are the largest and smallest are given by $\pgamma$ and $|\gamma_\bot\ra\la\gamma_\bot|$, where $|\gamma_\bot\ra$ is the orthogonal complement of $|\gamma\ra$. The imaginary components of the corresponding weak values are given by $\pm \tan\left(\frac{\omega}{2}\right)$, where $\omega$ is given as in Eq. \eref{Eq:Omega}.
\end{theorem}

It follows that the range of the imaginary components of the weak values are bounded by $\pm \tan\left(\frac{\omega}{2}\right)$ for the given PPS ensemble when we restrict the set of weak observables to state projectors. Again, as in the real case, as $|\varphi\ra$ and $|\psi\ra$ become closer to orthogonal, the bounds on the range of imaginary components approach $\pm \infty$. Also as in the real case, the projectors that maximize and minimize the imaginary component of the weak value correspond to states that are orthogonal to each other. Figure \ref{Fig:ImaginaryBloch} shows how finding the projectors that are farthest away from the PPS plane is equivalent to shifting the PPS plane up and down until it is tangent to the Bloch sphere. The tangent points then correspond to the operators whose imaginary component is maximized or minimized.

\begin{figure}
\centering
\includegraphics[width=\textwidth]{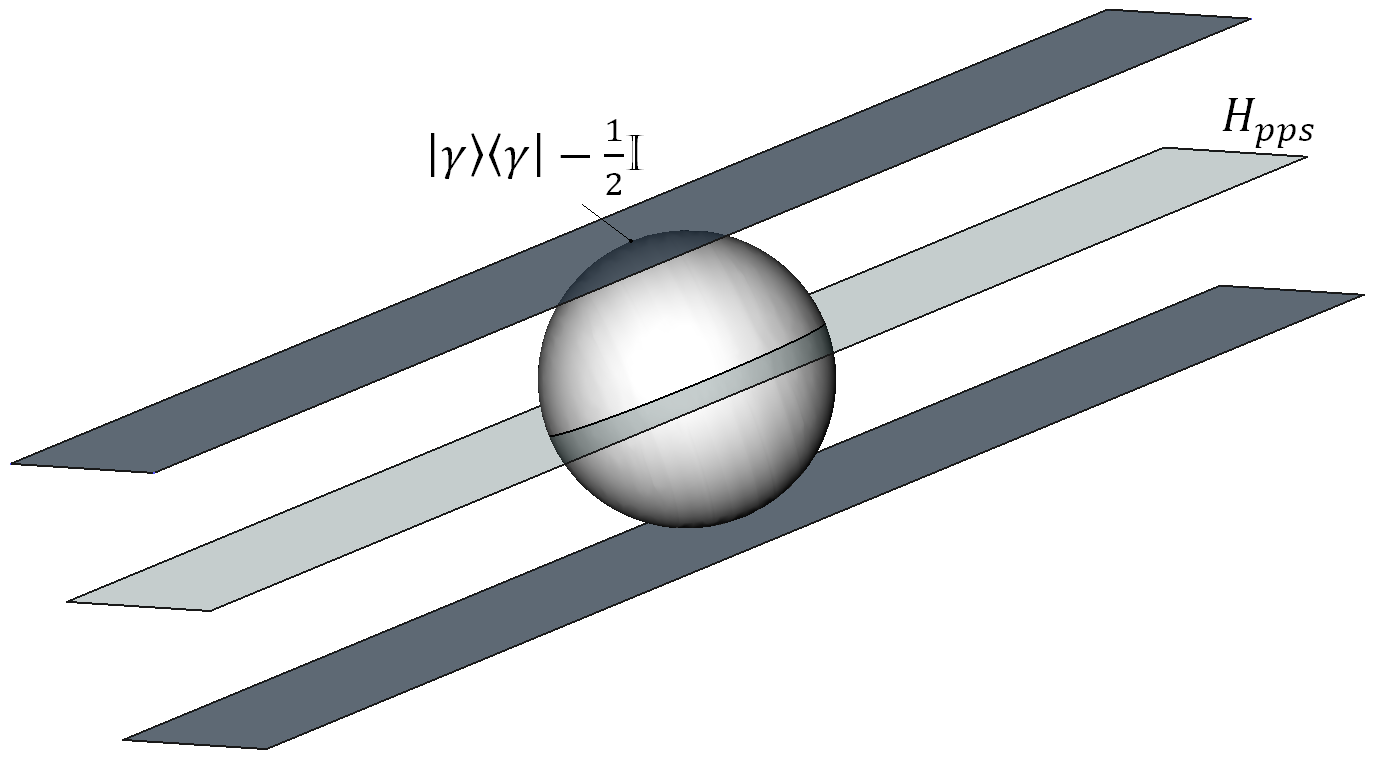}
\caption{If the PPS plane $H_{pps}$ is shifted up and down until it is tangent to the Bloch sphere, the tangent points are then the points that are farthest from the PPS plane on the Bloch sphere. These two points correspond precisely to the two states, $|\gamma\ra$ and $|\gamma_\bot\ra$, that are mutually unbiased to both states in the PPS ensemble.}
\label{Fig:ImaginaryBloch}
\end{figure}

\section{A Generalization}\label{Sec:Generalization}
Thus far, we have assumed that the pre- and post-selected states are both pure. However, this is often not the case in most practical implementations. While a state may begin as pure, by the time it arrives at the moment of weak measurement, it likely has undergone some form of noise. This should obviously affect the weak measurement results. We are interested in generalizing our definition of weak functions such that the pre-selected state may be a density operator. Doing so will help us understand how noise affects a weak value. In particular, if we begin with a designated pure state and we post-select on a designated pure state, then we can easily calculate what the expected weak value of a particular Hermitian operator \emph{should} be had no noise occurred. Comparing this expected weak value against an experimentally obtained weak value should reveal some information about the noise.

Because the weak value is conditioned on both the pre- and post-selected states, one may expect that noise occurring between the moment of weak measurement and the moment of post-selection may also affect the weak value. However, if we make the weak measurement immediately prior to post-selection, then any noise between weak measurement and post selection is essentially negligible. Consequently, all of the noise parameters will be captured in the pre-selected state.

If the pre-selected state is a density operator $\rho$ and the post-selected state is a pure state $|\psi\ra$, then using the generalized weak value formula of \cite{DJ12}, we may define the generalized weak function on the PPS ensemble $(\rho, |\psi\ra)$ as a complex linear operator
\begin{equation}
W_{\rho, \psi}: \mbb{H} \rightarrow \mbb{C}
\end{equation}
given by
\begin{equation}\label{Eq:GenWeak}
W_{\rho, \psi}(M): = \frac{\Tr(\ppsi M \rho)}{\Tr(\ppsi\rho)}.
\end{equation}
In this more general setting, the assumption that the PPS pair be distinct and non-orthogonal generalizes to the assumption that $\rho \not\in \aff\{\ppsi, |\psi_\bot\ra\la\psi_\bot|\}$.

Because $\rho$ is a density operator, we may always express it as a convex sum of projectors onto orthogonal pure states, that is, $\rho = (1-p)\pvarphi + p|\varphi_\bot\ra\la\varphi_\bot|$ for some $p \in [0,1]$ and some pure state $|\varphi\ra$ with orthogonal complement $|\varphi_\bot\ra$. Assuming $\rho \not\in \aff\{\ppsi, |\psi_\bot\ra\la\psi_\bot|\}$ is equivalent to assuming that $|\varphi\ra \not\in \{|\psi\ra, |\psi_\bot\ra\}$ and that $p \neq \frac{1}{2}$. If these conditions are true, then there exists a corresponding weak value function $W_{\rho, \psi}$, and for any $M \in \mbb{H}$, we may expand \eref{Eq:GenWeak} in the following way:
\begin{eqnarray}
 W_{\rho, \psi}(M) &= \frac{\Tr(|\psi\ra\la\psi|M\rho)}{\Tr(|\psi\ra\la\psi|\rho)} &\ \ \ \ \ \ \ \ \  \\
 \ &= \frac{\Tr\left[(1-p)\ppsi M\pvarphi + p\ppsi M|\varphi_\bot\ra\la\varphi_\bot|\right]}{\Tr(\ppsi\rho)} &\ \\
 \ &= \frac{(1-p)\la\psi|M\pvarphi\psi\ra + p\la\psi|M|\varphi_\bot\ra\la\varphi_\bot|\psi\ra}{\Tr(\ppsi\rho)} &\ \\ \label{Eq:WeakStatistical}
 \ &= \frac{(1-p)|\la\psi|\varphi\ra|^2}{\Tr(\ppsi\rho)}W_{\varphi, \psi}( M) + \frac{p|\la\psi|\varphi_\bot\ra|^2}{\Tr(\ppsi\rho)}W_{\varphi_\bot, \psi}( M). &\ 
\end{eqnarray}
In particular, when $p = 0$, then $\rho = \pvarphi$, and the generalized weak function reduces to the standard weak function formalism of \eref{Eq:WeakDef}.

We see then, that the generalized weak function can always be expressed as a statistical mixture of pure state weak functions. In fact, if $|\varphi\ra$ and $|\psi\ra$ are mutually unbiased, then so must be $|\psi\ra$ and $|\varphi_\bot\ra$, and the above significantly simplifies to
\begin{equation}\label{Eq:MutUnbiased}
 W_{\rho, \psi}(M) = (1-p) W_{\varphi, \psi}( M) + p W_{\varphi_\bot, \psi}( M).
\end{equation}

Geometrically, we can interpret quantum mechanical noise as some action on the Bloch ball. In particular, the effect of depolarizing noise is to uniformly shrink the Bloch ball towards its center. Consequently, the effect of depolarizing noise on density operators is to map a state projector $\pvarphi$ to the point $(1-p)\pvarphi + p|\varphi_\bot\ra\la\varphi_\bot|$ for some $p \in [0, \frac{1}{2})$. If we are trying to characterize the depolarizing noise, then the parameter we are trying to obtain is $p$. 

By \eref{Eq:WeakStatistical}, we see that, given any PPS ensemble $(|\varphi\ra, |\psi\ra)$, we may weakly measure essentially any Hermitian operator $M \in \mbb{H}$ (to obtain the value $W_{\rho, \psi}(M)$, where $\rho = (1-p)\pvarphi + p|\varphi_\bot\ra\la\varphi_\bot|$), and then by decomposing the experimentally obtained value as in \eref{Eq:WeakStatistical}, we may infer the value $p$. However, taking practical experimental concerns into consideration, there may be certain operators that are more preferable to measure than others. In particular, we would like to compare the observed weak value $W_{\rho, \psi}(M)$ of the operator $M$ against the expected weak value $W_{\varphi, \psi}(M)$ had there been no noise at all. Depolarizing noise always decreases the magnitude of the weak value. Due to finite sample size and resolution, this difference between observed and expected weak values becomes easier to experimentally detect the larger the magnitude of the expected weak value. If we are further restricted to only measuring the real or imaginary components of a weak value, and our weak measurement observable is a pure state projector, then the ideal choice of measurement observables are those given by Equation \eref{Eq:H_MaxMin} in the case that we are measuring the change in the real component, and $|\gamma\ra\la\gamma|$ and $|\gamma_\bot\ra\la\gamma_\bot|$ when we are measuring the change in imaginary component, where again $|\gamma\ra$ is a state mutually unbiased to both $|\varphi\ra$ and $|\psi\ra$. 

Inferring the value $p$ becomes exceptionally easier in the event that $|\varphi\ra$ and $|\psi\ra$ are mutually unbiased, as Equation \eref{Eq:MutUnbiased} shows. Simple calculations show that $W_{\varphi_\bot, \psi}(H_{\varphi, \psi}^\pm) = \frac{1}{2}$, so that in the mutually unbiased case, we obtain
\begin{eqnarray}
W_{\rho, \psi}(H^\pm) &= \frac{p}{2} + (1-p)W_{\varphi, \psi}(H_{\varphi, \psi}^\pm)\\
 \                    &= \frac{1}{2}\left( 1 \pm (1-p)\sqrt{2}\right).
\end{eqnarray}
Likewise, in the mutually unbiased case, we have $\omega = \pm\frac{\pi}{2}$, and hence weakly measuring $|\gamma\ra\la\gamma|$ returns
\begin{equation}
W_{\rho, \psi}(\pgamma) = \frac{1}{2}\left( 1 \pm (1-2p)i\right),
\end{equation}
so that the imaginary component is simply $\pm(\frac{1}{2} - p)$, where the sign is determined by the sign of $\omega$. Similar calculations show that $W_{\rho, \psi}(|\gamma_\bot\ra\la\gamma_\bot|)$ is the complex conjugate of $W_{\rho, \psi}(\pgamma)$, so that the imaginary component of $W_{\rho, \psi}(|\gamma_\bot\ra\la\gamma_\bot|)$ is simply the negative of that of $W_{\rho, \psi}(\pgamma)$.

In order to most easily experimentally distinguish between the expected (e.g. noiseless) and observed (e.g. noisy) weak values, we would like the difference between expected and calculated values to be as large as possible; that is, we would like $|\re(W_{\varphi, \psi}(M)) - \re(W_{\rho, \psi}(M))|$ or $|\im(W_{\varphi, \psi}(M)) - \im(W_{\rho, \psi}(M))|$ to be maximized. Interestingly, while it is the case that for all $p \in [0,1)$, we have $|\re(W_{\rho, \psi}(H^+))| > |\im(W_{\rho, \psi}(\pgamma))|$, it is also the case that for all $p \in [0, 1)$, the differences between the magnitudes of the observed and expected weak values satisfy
\begin{eqnarray*}
|\im(W_{\varphi, \psi}(\pgamma)) &- \im(W_{\rho, \psi}(\pgamma))|\\
        \ &\geq |\re(W_{\varphi, \psi}(H_{\varphi, \psi}^\pm)) - \re(W_{\rho, \psi}(H_{\varphi, \psi}^\pm))|,
\end{eqnarray*}
with equality only in the case of $p=0$. Consequently, $\pgamma$ (or equivalently $|\gamma_\bot\ra\la\gamma_\bot|$) is likely the more preferable observable to measure to characterize the noise in this scenario.


\section{Conclusion}\label{Sec:Conclude}
In conclusion, we have shown that the weak value of an observable conditioned on a PPS ensemble $(|\varphi\ra, |\psi\ra)$ of pure states can be viewed mathematically as the image of a well-defined linear function $W_{\varphi, \psi}$ from the vector space $\mbb{H}$ of Hermitian matrices to the set $\mbb{C}$ of complex numbers. Because of the inherent Euclidean structure underlying the Hermitian vector space, we can use this map to characterize geometric invariants of weak values. In particular, when restricting to qubit weak measurements over the subspace $\mc{S}$ of trace-0 Hermitian matrices, each weak value (i.e. each complex number) corresponds to a unique affine line in $\mc{S}$. This geometric characterization, at least when restricted to qubits, has great utility for, e.g. finding the weak projectors that maximize the real or imaginary components of the weak value, given a particular PPS ensemble. Moreover, we have shown that the weak value of an observable $M \in \mbb{H}$ can be decomposed into different geometric components. In particular, we may always write
\begin{equation}
W_{\varphi, \psi}(M) = \frac{1}{2} \left(\Tr(M) + s +ia\tan\left(\frac{\omega}{2}\right)\right),
\end{equation}
where $\Tr(M)$ indicates in which affine shift of $\mc{S}$ the operator $M$ lies, $s$ corresponds to the invariant line on which the image of $M$ in $\mc{S}$ lies, and $a$ corresponds to the distance of the image of $M$ in $\mc{S}$ from the PPS plane $H_{pps}$. The term $\omega$, however, is determined by the weak function itself. While it can be characterized as a geometric relation between the pre- and post-selected states (as in Eq. \eref{Eq:Omega}), it has no relation to the choice of operator $M$.

Additionally, we have shown a straightforward way to generalize the weak value functions such that the pre-selected system in the PPS ensemble defining the function may be a density operator. Doing this provides a method to characterize large classes of channel noise via the attenuation of the weak value. We demonstrated how to explicitly calculate the effects of depolarizing noise when the initial state (prior to noise) and post-selected state were mutually unbiased, and use the geometric methods devised in this paper to determine the optimal state projectors to weakly measure. The same methods can be adapted to characterize a variety of noise classes. For example, the amplitude damping channel maps an arbitrary state projector $\pvarphi$ to $(1-p)\pvarphi + p|0\ra\la 0|$, for some $p \in [0, 1)$, where $|0\ra$ is generally taken to be the $+1$-eigenstate of the Pauli spin-$Z$ operator. In order to fully characterize the amplitude damping channel, we need to determine the value $p$. With only slight modifications, we may use the same methods to characterize amplitude damping channel noise as well.

In a prepare-and-measure quantum information processing task, the preparation and measurement bases are generally pre-determined. By recasting weak measurements and weak values in terms of linear functions defined by the pre- and post-selected states, we have made it particularly useful to analyze the effects of weak measurements in existing prepare-and-measure quantum information tasks. For example, we can envision augmenting the BB84 QKD protocol so that the receiver makes weak measurements immediately prior to strongly measuring in a particular basis. Given the possible PPS ensembles in BB84, we can use the techniques in this paper to easily determine which observables to weakly measure to obtain the maximum information about the quantum channel. When the bases disagree, we can use the attenuation in the corresponding weak values to characterize the channel noise, and consequently perform security analysis. Such an augmented BB84 QKD protocol was recently proposed in \cite{TF15}.

Admittedly, there are obvious practical shortfalls in using the generalized weak values to characterize noise. In particular, we have implicitly assumed perfect detectors, and we have arranged the experiments such that the weak measurement occurs immediately before post-selection so that any noise that might occur between the weak measurement and post selection can be treated as negligible. Indeed, even the best single photon detectors will have some nonzero dark count rate. Moreover, while a weak measurement does not collapse the quantum system, it will introduce a very marginal amount of disturbance which will show up as an extremely small, yet nonzero, error factor between the weak measurement and post-selection. Generalizing the weak value functions further such that both elements of the PPS ensemble are density operators will be a topic addressed in future work. We note that \cite{DJ12} provides a characterization of weak measurements in which the post-selection can be the result of any POVM. Using this characterization of weak measurements, we should be able to obtain the desired generalization.

Another obvious next step is to generalize all of these results to higher-dimensional quantum systems. While we have defined weak functions over $n \times n$ Hermitian matrices for any integer $n \geq 2$, we have thus far only characterized the geometric aspects of these functions in the case of $n = 2$. Generalizing to the case of $n > 2$ is an obvious, yet daunting next step, given that the dimension of $\mbb{H}$ grows quadratically with $n$. Even determining necessary and sufficient conditions for a weak measurement to yield an entirely real weak value in this case is highly nontrivial, though we have some conjectures.

In particular, suppose $|\varphi_0\ra$ and $|\psi_0\ra$ are distinct, nonorthogonal \emph{qunits} (that is, they exist as unit-norm vectors in $\mbb{C}^n$). Suppose $\{|\varphi_0\ra, |\varphi_1\ra, \dots, |\varphi_{n-1}\ra\}$ and $\{|\psi_0\ra, |\psi_1\ra, \dots, |\psi_{n-1}\ra\}$ are each orthonormal bases containing $|\varphi_0\ra$ and $|\psi_0\ra$, respectively. Let $\mc{P}_\varphi = \{ |\varphi_0\ra\la\varphi_0| - \frac{1}{n}\mbb{I}, \dots, |\varphi_{n-1}\ra\la\varphi_{n-1}| - \frac{1}{n}\mbb{I}\}$, and define $\mc{P}_\psi$ likewise. That is, the elements of $\mc{P}_\varphi$ and $\mc{P}_\psi$ are the images of the projectors onto the corresponding basis elements in the trace-0 subspace $\mc{S}$ of $\mbb{H}$. Let $\mc{R}_{\varphi, \psi} = \Span_\mc{S}\{\mc{P}_\varphi, \mc{P}_\psi\}$. Then straightforward calculations show the following result.

\begin{proposition}
Let $M \in \mc{R}_{\varphi, \psi}$. Then $\mc{W}_{\varphi_i, \psi_j}(M) \in \mbb{R}$, for any $i,j \in \{0, 1, \dots n-1\}$.
\end{proposition}
In the case $n = 2$, the space $\mc{R}_{\varphi, \psi}$ is precisely the PPS plane $H_{pps}$. Consequently, we conjecture the following:
\begin{conjecture}
Let $M \in \mc{S}$. Then for all $i,j \in \{0, 1, \dots n-1\}$, $\mc{W}_{\varphi_i, \psi_j}(M) \in \mbb{R}$ if and only if $M \in \mc{R}_{\varphi, \psi}$.
\end{conjecture}
Determining whether or not the above conjecture is true would be the first step in characterizing the geometry of weak measurements in higher dimensions.

\ack

J. Farinholt acknowledges support from the Quantum Information Science Program, Office of Naval Research Code 31. A. Ghazarians acknowledges support from the Office of Naval Research, Warfare Innovation Cell Knowledge and Education (WICKED), which funded his internship at NSWCDD. J. Troupe acknowledges support from the Office of Naval Research, grant number N00014-15-1-2225.



\section*{References}
\bibliographystyle{unsrt}
\bibliography{bibfile3}

\begin{thebibliography}{10}

\bibitem{AV90}
Yakir Aharonov and Lev Vaidman.
\newblock Properties of a quantum system during the time interval between two
  measurements.
\newblock {\em Phys. Rev. A}, 41:11--20, Jan 1990.

\bibitem{Dixon09}
P.~Ben Dixon, David~J. Starling, Andrew~N. Jordan, and John~C. Howell.
\newblock Ultrasensitive beam deflection measurement via weak value
  amplification.
\newblock {\em Physical Review Letters}, 102:173601, 2009.

\bibitem{Lyons15}
Kevin Lyons, Justin Dressel, Andrew~N. Jordan, John~C. Howell, and Paul~G.
  Kwiat.
\newblock Power-recycled weak-value based metrology.
\newblock {\em Physical Review Letters}, 114:170801, 2015.

\bibitem{Viza15}
Geraoldo~I. Viza, Julian Martinez-Ricon, Gabriel~B. Alves, Andrew~N. Jordan,
  and John~C. Howell.
\newblock Experimentally quantifying the technical advantages of
  weak-value-based metrology.
\newblock {\em Physical Review A}, 92:032127, 2015.

\bibitem{Hosten08}
Onur Hosten and Paul Kwiat.
\newblock Observation of the spin hall effect oflight via weak measurements.
\newblock {\em Science}, 319:787 -- 790, 2008.

\bibitem{PuseyLeifer15}
Matthew~F. Pusey and Matthew~S. Leifer.
\newblock Logical pre- and post-selection paradoxes are proofs of
  contextuality.
\newblock {\em 12th International Workshop on Quantum Physics and Logic (QPL
  2015)}, 195:295 -- 306, 2015.

\bibitem{Pus14}
Matthew~F. Pusey.
\newblock Anomalous weak values are proofs of contextuality.
\newblock {\em Phys. Rev. Lett.}, 113:200401, Nov 2014.

\bibitem{TF15}
James Troupe and Jacob Farinholt.
\newblock A contextuality based quantum key distribution protocol.
\newblock {\em arXiv:1512.02256 [quant-ph]}.

\bibitem{Delfosse15}
Nicolas Delfosse, Phillipe~A. Guerin, Jacob Bian, and Robert Rausendorf.
\newblock Wigner function negativity and contextuality in quantum computation
  on rebits.
\newblock {\em Physical Review X}, 5:021003, 2015.

\bibitem{Howard14}
Mark Howard, Joel Wallman, Victor Veitch, and Joseph Emerson.
\newblock Contextuality supplies the 'magic' for quantum computation.
\newblock {\em Nature}, 510:351--355, 2014.

\bibitem{Lobo14}
Augusto~C. Lobo, Yakir Aharonov, Jeff Tollaksen, Elizabeth~M. Berrigan, and
  Clyffe de~Assis~Ribeiro.
\newblock Weak values and modular variables from a quantum phase-space
  perspective.
\newblock {\em Quantum Studies: Mathematics and Foundations}, 1(1-2):97--132,
  2014.

\bibitem{Tam09}
S~Tamate, H~Kobayashi, T~Nakanishi, K~Sugiyama, and M~Kitano.
\newblock Geometrical aspects of weak measurements and quantum erasers.
\newblock {\em New Journal of Physics}, 11(9):093025, 2009.

\bibitem{DJ12}
J.~Dressel and A.~N. Jordan.
\newblock Contextual-value approach to the generalized measurement of
  observables.
\newblock {\em Phys. Rev. A}, 85:022123, Feb 2012.

\end{thebibliography}

\end{document}